\documentclass[11pt]{article}
\usepackage[left=1in,top=1in,right=1in,bottom=1in,head=.1in,nofoot]{geometry}

\usepackage{setspace,url,bm,amsmath} 
\usepackage{authblk}

\usepackage{titlesec} 
\titlelabel{\thetitle.\quad} 
\titleformat*{\section}{\bf\large\center\uppercase} 

\usepackage{graphicx} 
\usepackage{bbm}
\usepackage{latexsym}
\usepackage{hyperref}

\usepackage{amsthm}
\usepackage{amsfonts}
\theoremstyle{definition}
\newtheorem{assumption}{Assumption}

\newtheorem{theorem}{Theorem}
\newtheorem{lemma}{Lemma}
\newtheorem{example}{Example}

\newtheorem{corollary}{Corollary}

\usepackage{natbib} 
\bibpunct{(}{)}{;}{a}{}{,} 

\usepackage{etoolbox} 
\apptocmd{\sloppy}{\hbadness 10000\relax}{}{} 

\begin{document}
\doublespacing
\title{\bf Covariate adjustment in randomization-based causal inference for $2^K$ factorial designs}
\author{Jiannan Lu\footnote{Address for correspondence: Jiannan Lu, Microsoft Corporation, One Microsoft Way, Redmond, Washington 98052, USA.
Email: \texttt{jiannl@microsoft.com}} \\ Analysis and Experimentation, Microsoft Corporation}

\date{}
\maketitle
\begin{abstract}
We develop finite-population asymptotic theory for covariate adjustment in randomization-based causal inference for $2^K$ factorial designs. In particular, we confirm that both the unadjusted and the covariate-adjusted estimators of the factorial effects are asymptotically unbiased and normal, and the latter is more precise than the former.
\end{abstract}

\textbf{Keywords:} Potential outcome; Variance reduction; Finite-population asymptotics

\section{Introduction}\label{sec:intro}

Randomization is often considered the gold standard for causal inference \citep{Rubin:2008}. A well-established methodology to conduct causal inference is the potential outcomes framework \citep{Neyman:1923, Rubin:1974}, which defines the causal effect of a binary treatment factor as the comparison between the potential outcomes under treatment and control. In the presence of multiple binary treatment factors, we can evaluate them simultaneously under the $2^K$ factorial design framework \citep{Fisher:1935, Yates:1937}. Several researchers \citep[e.g.,][]{Kempthrone:1952, Kempthrone:1955, Wilk:1956, Bailey:1981, Bailey:1991, Dasgupta:2015} advocated conducting randomization-based causal inference for $2^K$ factorial designs, which has several advantages over the widely-used regression-based inference. For example, randomization-based inference is applicable to the finite-population setting, and therefore may be more reasonable in practice \citep[e.g.,][]{Miller:2006, Lu:2015}. For more discussion on the comparison and reconciliation of randomization-based and regression-based inferences for $2^K$ factorial designs, see \cite{Lu:2016}.

In randomization-based causal inference, covariate adjustment \citep{Cochran:1977} is a variance reduction technique widely used by researchers \citep[e.g.,][]{Deng:2013, Miratrix:2013}. In an illuminating paper, \cite{Lin:2013} demonstrated the advantages of performing covariate adjustment for randomized treatment-control studies (i.e., $2^1$ factorial designs). However, to our best knowledge, for $2^K$ factorial designs which are of great importance from both theoretical and practical perspectives, similar discussions appear to be absent; it is unclear whether covariate adjustment is beneficial for $2^K$ factorial designs, and if so, how to quantify said benefit. In this paper we answer this question, by extending the discussions in \cite{Lin:2013} and illustrating the advantages of performing covariate adjustment in $2^K$ factorial designs. To be specific, we derive the closed-form expressions for the asymptotic precisions of the unadjusted and covariate-adjusted estimators, and thus accurately measure the precision gained by covariate adjustment.

The paper proceeds as follows. Section \ref{sec:2k} reviews randomization-based inference for $2^K$ factorial designs. Section \ref{sec:def} introduces the covariate-adjusted estimator for $2^K$ factorial designs. Section \ref{sec:adj} derives the asymptotic precisions of the unadjusted and covariate-adjusted estimators. Section \ref{sec:discuss} concludes and discusses possible future directions.

\section{Randomization Inference for $2^K$ Factorial Designs}\label{sec:2k}

In this section, we review the randomization-based inference framework for $2^K$ factorial designs \citep{Dasgupta:2015, Lu:2016}. For consistency we adopt the notations in \cite{Lu:2016}.

\subsection{$2^K$ factorial designs}

$2^K$ factorial designs consist of $K$ distinct treatment factors, each of which has two levels coded as -1 and 1. To simplify future notations we let $J=2^K.$ To define $2^K$ factorial designs, we rely on a $J \times J$ orthogonal matrix $\bm H = (\bm h_0, \ldots, \bm h_{J-1}),$ whcih is often referred to as the the model matrix \citep{Wu:2009}. We construct the model matrix in the following recursive way \citep{Espinosa:2016, Lu:2016}: 
\begin{enumerate}
\item Let $\bm h_0 = \bm 1_J;$ 

\item For $k=1,\ldots,K$, construct $\bm h_k$ by letting its first $2^{K-k}$ entries be -1, the next $2^{K-k}$ entries be 1, and repeating $2^{k-1}$ times;

\item If $K \ge 2,$ order all subsets of $\{1, \ldots, K\}$ with at least two elements, first by cardinality and then lexicography. For $k = 1, \ldots J-1-K,$ let $\sigma_k$ be the $k$th subset and $\bm h_{K+k} = \prod_{l \in \sigma_k} \bm h_l,$ where ``$\prod$'' stands for entry-wise product.

\end{enumerate}
The $j$th row of the sub-matrix $\tilde{\bm H} = (\bm h_1, \ldots, \bm h_K)$ is the $j$th treatment combination $\bm z_j.$ To further illustrate the construction of the model matrix, we adopt the example in \cite{Lu:2016}.

\begin{example}
Let $K=2.$ By following the above recursive procedure, we obtain 
$
\bm h_0 = \bm 1,
$
$
\bm h_1 = (-1, -1, 1, 1)^\prime,
$
$
\bm h_2=(-1, 1, -1, 1)^\prime,
$
and
$
\bm h_3=(1, -1, -1, 1)^\prime.
$
Consequently, for $2^2$ factorial designs the model matrix is:
\begin{equation*}
\bm H =
\bordermatrix{& \bm h_0 & \bm h_1& \bm h_2 & \bm h_3\cr
              & 1  & -1 &  -1  & 1 \cr
              & 1 & -1 &  1  & -1 \cr
              & 1  & 1 &  -1  & -1 \cr
              & 1 & 1 & 1 & 1}.
\end{equation*}
The four treatment combinations are $\bm z_1=(-1, -1),$ $\bm z_2=(-1, 1),$ $\bm z_3=(1, -1)$ and $\bm z_4=(1, 1).$
\end{example}

\subsection{Randomization-based Inference}

We allow $N \ge 2J$ experimental units in the design. To describe the randomization-based inference framework, we follow a three-step procedure.

First, under the Stable Unit Treatment Value Assumption \citep{Rubin:1980} that for $j=1, \ldots, J$ there is only one version of the treatment combination $\bm z_j,$ and no interference among the experimental units, let $Y_i(\bm z_j)$ be the potential outcome of unit $i$ under treatment combination $\bm z_j,$ and
$
\bar Y(\bm z_j) = N^{-1} \sum_{i=1}^N Y_i(\bm z_j)
$
be the average potential outcome across all the experimental units. Let 
$
\bm Y_i = \{ Y_i(\bm z_1), \ldots, Y_i(\bm z_J) \}^\prime
$
and
$
\bar{\bm Y} = \{ \bar Y(\bm z_1), \ldots, \bar Y(\bm z_J) \}^\prime.
$

Next, we randomly assign $n_j \ge 2$ units to treatment combination $\bm z_j.$ Let
\begin{equation*}
W_i(\bm z_j)
= 
\begin{cases}
1, & \text{if unit $i$ is assigned treatment $\bm z_j,$ } \\
0, & \text{otherwise,} \\
\end{cases}
\end{equation*}
and let
$
Y_i^\textrm{obs} = \sum_{j=1}^J W_i(\bm z_j) Y_i(\bm z_j)
$
be the observed outcome for unit $i,$ and therefore the average observed outcome across all experimental units that are assigned to treatment combination $\bm z_j$ is
$
\bar Y^\textrm{obs}(\bm z_j) = n_j^{-1} \sum_{i=1}^N W_i(\bm z_j)Y_i(\bm z_j).
$
Furthermore, we let
$
\bar{\bm Y}^\textrm{obs} = \{ \bar Y^\textrm{obs}(\bm z_1) , \ldots, \bar Y^\textrm{obs}(\bm z_J) \}^\prime.
$

Finally, we define the factorial effects as
\begin{equation*}
\tau(l) = \frac{1}{2^{K-1}} \bm h_l^\prime \bar{\bm Y}
\quad
(l = 1, \ldots, J-1),
\end{equation*}
and their randomization-based estimators as
\begin{equation}\label{eq:est}
\hat \tau_\textrm{rb}(l) =  \frac{1}{2^{K-1}} \bm h_l^\prime \bar{\bm Y}^\textrm{obs}
\quad
(l = 1, \ldots, J-1).
\end{equation}
Its randomness is solely from the treatment assignment $W_i(\bm z_j)$'s.

\section{Covariate Adjustment in $2^K$ Factorial Designs}\label{sec:def}
 
The idea behind the randomization-based estimator is estimating the average potential outcome $\bar Y(\bm z_j)$ by its corresponding average observed outcome $\bar Y^\textrm{obs}(\bm z_j).$ However, as shown in \cite{Cochran:1977} and later mentioned in \cite{Lin:2013}, utilizing the pre-treatment covariates can potentially improve the precision of $\bar Y^\textrm{obs}(\bm z_j),$ and consequently that of the randomization-based estimator. With this classic wisdom, we define the covariate-adjusted estimator for $2^K$ factorial designs. In this paper, we consider the method of covariate adjustment where separate slope coefficients are estimated for each average potential outcome $\bar Y(\bm z_j),$ unlike the traditional ANCOVA method in which there is only one pooled slope coefficient. The rationale behind this is from the existing literature on covariate adjustment in randomized treatment-control studies -- as shown in \cite{Freedman:2008} and \cite{Lin:2013}, the traditional ANCOVA can potentially help or hurt asymptotic precision, however the ``separate slope'' method guarantees asymptotic precision improvement.

Let $\bm X_i = (X_{i1}, \ldots, X_{ip} )^\prime$ be the pre-treatment covariates of the unit $i,$ and  
$
\bar X_k = N^{-1} \sum_{i=1}^N X_{ik}
$
and
$
\bar X_k^\textrm{obs}(\bm z_j) = n_j^{-1} \sum_{i=1}^N W_i(\bm z_j) X_{ik}
$
be the average of the $k$th covariate of all units and those assigned to treatment $\bm z_j.$ Let
$
\bar {\bm X} = (\bar X_1, \ldots, \bar X_p )^\prime,
$
and
$
\bar {\bm X}^\textrm{obs}(\bm z_j) = 
\{
\bar X_1^\textrm{obs}(\bm z_j), \ldots, \bar X_p^\textrm{obs}(\bm z_j) 
\}^\prime.
$
Consider the following type of estimators for $\bar Y(\bm z_j):$
\begin{equation*}
\bar Y^\textrm{obs}(\bm z_j) +  \{ \bar{\bm X} - \bar{\bm X}^\textrm{obs}(\bm z_j) \}^\prime \bm \beta_j,
\end{equation*}
where $\bm \beta_j$ is a constant vector to be determined. As shown in \cite{Cochran:1977}, the value of $\bm \beta_j$ that minimizes the variance of the above is
\begin{equation}\label{eq:true-reg-coef}
\bm \beta_j = \left\{ \frac{1}{N} \sum_{i=1}^N (\bm X_i - \bar{\bm X}) (\bm X_i - \bar{\bm X})^\prime \right\}^{-1} \left[ \frac{1}{N} \sum_{i=1}^N (\bm X_i - \bar{\bm X}) \{Y_i(\bm z_j) - \bar Y(\bm z_j)\} \right],
\end{equation}
which we assume to be well-defined, i.e., the ``design matrix''
$
N^{-1} \sum_{i=1}^N (\bm X_i - \bar{\bm X}) (\bm X_i - \bar{\bm X})^\prime    
$
is invertible. We estimate \eqref{eq:true-reg-coef} by the plug-in method:
\begin{equation}\label{eq:est-reg-coef}
\hat{\bm \beta}_j = \left\{ \frac{1}{N} \sum_{i=1}^N (\bm X_i - \bar{\bm X}) (\bm X_i - \bar{\bm X})^\prime \right\}^{-1} \left[ \frac{1}{n_j} \sum_{i=1}^N W_i(\bm z_j)(\bm X_i - \bar{\bm X}) \{Y_i(\bm z_j) - \bar Y^\textrm{obs}(\bm z_j)\} \right],
\end{equation}
and let
\begin{equation}\label{eq:adj-single}
\bar Y^\textrm{ca}(\bm z_j) = \bar Y^\textrm{obs}(\bm z_j) +  \{ \bar{\bm X} - \bar{\bm X}^\textrm{obs}(\bm z_j) \}^\prime \hat{\bm \beta}_j
\quad
(j = 1, \ldots, J).
\end{equation}
Consequently, we define the covariate-adjusted estimator as
\begin{equation}\label{eq:adj}
\hat \tau_\textrm{ca}(l) 
=  
\frac{1}{2^{K-1}} {\bm h}_l^\prime \bar {\bm Y}^\textrm{ca}
\quad
(l = 1, \ldots, J - 1),
\end{equation}
where $\bar {\bm Y}^\textrm{ca} = \{ \bar Y^\textrm{ca}(\bm z_1), \ldots, \bar Y^\textrm{ca}(\bm z_J) \}^\prime.$

\section{Finite-Population Asymptotic Analysis}\label{sec:adj}

\subsection{Notations and Assumptions}

Consider a hypothetical sequence of finite populations with increasing sample sizes. Technically, all the finite-population quantities should have superscripts that index the sequence of populations, for example
$
\bar{ \bm X }^{(N)} = N^{-1} \sum_{i=1}^N  \bm X_i.
$
For convenience we drop all superscripts. We make the following assumptions to conduct the finite-population asymptotic analysis, and the first assumption merely serves the purpose of simplifying notations.

\begin{assumption}\label{assum:0}
For all $N,$ the potential outcomes and covariates are centered at zeros, i.e.,
$
\bar { \bm X } =  0, 
$
and
$
\bar{ \bm Y } =  0.
$
\end{assumption}

\begin{assumption}\label{assum:1}
When $N \rightarrow \infty,$ the proportions of units assigned to all treatment combinations converge to positive constants, i.e.,
\begin{equation*}
\hat p_j = \frac{n_j}{N} \rightarrow p_j
\quad
(j=1, \ldots, J),
\end{equation*} 
where $p_j > 0$ for all $j$ and $\sum_{j=1}^J p_j = 1.$
\end{assumption}

\begin{assumption}\label{assum:2}
When $N \rightarrow \infty,$ all the second moments of the potential outcomes converge to constants, i.e.,
\begin{equation*}
\frac{1}{N} \sum_{i=1}^N  \bm Y_i \bm Y_i^\prime 
\rightarrow 
\bm \Sigma = (\sigma_{jj^\prime})_{1 \le j, j^\prime \le J}
\end{equation*}
where $\sigma_{jj} > 0$ for all $j.$ All the second moments of the covariates converge to constants, i.e.,
\begin{equation*}
\frac{1}{N} \sum_{i=1}^N  \bm X_i \bm X_i^\prime 
\rightarrow
\bm \Omega = (\omega_{ll^\prime})_{1 \le l,l^\prime \le p},
\end{equation*}
where $\bm \Omega$ is an invertible matrix. All the mixed second moments of the potential outcomes and the covariates converge to constants, i.e.,
\begin{equation*}
\frac{1}{N} \sum_{i=1}^N  \bm X_i Y_i( \bm z_j)\rightarrow  \bm \lambda_j
\quad
(j = 1, \ldots, J).
\end{equation*}
\end{assumption}

\begin{assumption}\label{assum:3}
For all $N,$ the fourth moments of the potential outcomes and the covariates are uniformly bounded from above by a positive constant, i.e.,
\begin{equation*}
\frac{1}{N} \sum_{i=1}^N Y^4_i(\bm z_j) \leq L
\quad
(j=1, \ldots, J);
\quad
\frac{1}{N} \sum_{i=1}^N X_{ik}^4 \leq L
\quad
(k = 1, \ldots, p).
\end{equation*}
\end{assumption}

We introduce several useful notations before moving forward. Let
\begin{equation*}
\bm \zeta_j = \bm \Omega^{-1} \bm \lambda_j
\quad
(j=1, \ldots, J);
\quad
R_i( \bm z_j) = Y_i( \bm z_j ) -  \bm X_i^\prime \bm \zeta_j
\quad
(i=1, \ldots, N),
\end{equation*}
and $ \bm R_i = \{ R_i( z_1), \ldots, R_i( z_J) \}^\prime.$ Consequently, 
\begin{equation}\label{eq:po-adj-second-moments}
\frac{1}{N} \sum_{i=1}^N  \bm R_i \bm R_i^\prime 
\rightarrow 
\tilde { \bm \Sigma} = (\tilde \sigma_{jj^\prime})_{1 \le j,j^\prime \le J},
\end{equation}
where
$
\tilde \sigma_{jj^\prime} 
= \sigma_{jj^\prime} -  \bm \lambda_j^\prime  \bm \Omega^{-1}  \bm \lambda_{j^\prime}.
$

\subsection{Useful Lemmas}

For finite-population asymptotic analysis of the randomization-based estimator and the covariate-adjusted estimator, we rely on the following lemmas, which are also of independent interests. The first lemma is the Combinatorial Central Limit Theorem from \cite{Hoeffding:1951}.

\begin{lemma}
\label{lemma:0} 
For fixed $N \in \mathbb Z^+,$ and $N^2$ constants
$
a_{jk}
$
$
(j,k = 1, \ldots, N),
$
let
\begin{equation}\label{eq:lemma-0-0}
b_{jk} 
= a_{jk} - \frac{1}{N} \sum_{j^\prime = 1}^N a_{j^\prime k} - \frac{1}{N} \sum_{k^\prime = 1}^N a_{jk^\prime} + \frac{1}{N^2} \sum_{j^\prime = 1}^N \sum_{k^\prime = 1}^N a_{j^\prime k^\prime}.
\end{equation}
Furthermore, let $(\nu_1, \ldots, \nu_N)$ be a random permutation of $(1, \ldots, N)$ and
$
S = \sum_{j=1}^N a_{j, \nu_j}.  
$
If
\begin{equation}\label{eq:lemma-0-1}
\lim_{N\rightarrow \infty}
\max_{1 \le j,k \le N} b^2_{jk}
\big /
\left(
\frac{1}{N} \sum_{j = 1}^N \sum_{k = 1}^N b^2_{jk} 
\right)
= 0,
\end{equation}
then when $N \rightarrow \infty,$
\begin{equation*}
\frac{S - E(S)}{\{\mathrm{Var} (S)\}^{1/2}}
\stackrel{\mathbb D}{\longrightarrow}
N(0, 1).
\end{equation*}
\end{lemma}

The second lemma is essentially the ``$2^K$ factorial design version'' of the multivariate finite-population Central Limit Theorem in \cite{Freedman:2008}. However, we provide a rigorous proof in this paper, where \cite{Freedman:2008} did not. 
 
\begin{lemma}
\label{lemma:1} 
When $N \rightarrow \infty,$
\begin{equation*}
N^{1/2} \bar{ \bm Y }^\textrm{obs} 
\stackrel{\mathbb D}{\longrightarrow} 
N( 0,  \bm \Sigma^\textrm{obs}),
\end{equation*}
where
\begin{equation*}
\bm \Sigma^\textrm{obs} = 
\begin{bmatrix}
    \frac{1-p_1}{p_1} \sigma_{11} & -\sigma_{12} & \dots  & -\sigma_{1 J} \\
   -\sigma_{21} & \frac{1-p_2}{p_2} \sigma_{22} & \ldots & -\sigma_{2 J} \\
    \vdots  & \vdots & \ddots & \ldots \\
    -\sigma_{J 1}  & \dots  & \ldots & \frac{1-p_J}{p_J} \sigma_{JJ} \\
\end{bmatrix}.
\end{equation*}
\end{lemma}

\begin{proof}
By Cramer-Wold theorem, we only need to prove that
\begin{equation}\label{eq:normal}
N^{1/2}  \bm t^\prime \bar{ \bm Y }^\textrm{obs} 
\stackrel{\mathbb D}{\longrightarrow}
N(0,  \bm t^\prime \bm \Sigma^\textrm{obs}  \bm t)
\end{equation}
for all $\bm t = (t_1, \ldots, t_J)^\prime \in \mathbb R^J.$ If $\bm t = \bm 0_J,$ \eqref{eq:normal} holds trivially. Otherwise $\bm t \in \mathbb R^J \backslash \{ \bm 0_J \}:$ 

First, by simple probability argument \cite[e.g.,][Lemma 1]{Lu:2016}, 
\begin{equation}\label{eq:lemma-1-5}
E(\bar{ \bm Y }^\textrm{obs}) =  \bm 0_J;
\quad
\textrm{Var} \{ \bar Y^\textrm{obs} ( \bm z_j) \} = \frac{1-\hat p_j}{\hat p_j} \frac{1}{N(N-1)} \sum_{i=1}^N Y_i^2( \bm z_j)
\quad
(j=1, \ldots, J),
\end{equation}
and
\begin{equation*}
\textrm{Cov} \{ \bar Y^\textrm{obs} ( \bm z_j), \bar Y^\textrm{obs} ( \bm z_{j^\prime}) \} = - \frac{1}{N(N-1)} \sum_{i=1}^N Y_i( \bm z_j)Y_i( \bm z_{j^\prime})
\quad
(j \ne j^\prime).
\end{equation*}
Therefore when $N \rightarrow \infty,$
\begin{equation}\label{eq:mean-variance}
E (N^{1/2}  \bm t^\prime \bar{ \bm Y }^\textrm{obs}) = 0,
\quad
\textrm{Var} (N^{1/2} \bm t^\prime \bar{ \bm Y }^\textrm{obs}) \rightarrow  \bm t^\prime  \bm \Sigma^\textrm{obs}  \bm t,
\end{equation}

Next, we prove that
\begin{equation}\label{eq:lemma-1-0}
\lim_{N \rightarrow \infty}\frac{\max_{1 \le i \le N} Y^2_i( \bm z_j)}{\sum_{i=1}^N Y^2_i( \bm z_j)} = 0
\quad
(j=1, \ldots, J).
\end{equation}
Let
\begin{equation*}
\eta_i = Y^2_i( \bm z_j) / \left\{ \sum_{i=1}^N Y_i^2( \bm z_j) \right\}
\quad
(i=1, \ldots, N),
\end{equation*}
and obviously $\sum_{i=1}^N \eta_i = 1.$ Furthermore, let $\eta = \max_{1 \le i \le N}\eta_i,$ and consequently
\begin{equation*}
\eta \leq \left( \sum_{i=1}^N \eta_i^2 \right)^{1/2} = N^{-1/2}
\left\{ 
\frac{1}{N}\sum_{i=1}^N Y_i^4( \bm z_j)
\right\}^{1/2}
\big ./
\left\{
\frac{1}{N} \sum_{i=1}^N Y_i^2( \bm z_j)
\right\}.
\end{equation*}
Therefore by Assumptions \ref{assum:2} and \ref{assum:3}
\begin{equation*}
\limsup_{N\rightarrow\infty} N^{1/2} \eta \le L^{1/2}/\sigma_{jj},
\end{equation*}
which implies \eqref{eq:lemma-1-0}.

Then, we adopt the notations in Lemma \ref{lemma:0} and let
\begin{equation}\label{eq:lemma-1-1}
a_{gi} = N^{1/2}
\begin{cases}
t_1Y_i( z_1)/n_1, & \text{for } 1 \le g \le n_1,\\
 & \vdots \\
t_JY_i( z_J)/n_J, & \text{for } \sum_{j=1}^{J-1} n_j + 1 \le g \le N.
\end{cases}
\quad
(i=1, \ldots, N),
\end{equation}
which implies that
\begin{equation*}
\sum_{g=1}^N a_{g, \nu_g} = N^{1/2}  \bm t^\prime \bar{ \bm Y }^\textrm{obs}. 
\end{equation*}
By Assumption \ref{assum:0}
\begin{equation*}
\sum_{i^\prime = 1}^N a_{gi^\prime} = 0,
\quad
(g=1, \ldots, N).
\end{equation*}
Therefore, if 
\begin{equation*}
\sum_{j^\prime = 1}^{j-1} n_{j^\prime} < g \le \sum_{j^\prime = 1}^j n_{j^\prime},
\end{equation*}
then by \eqref{eq:lemma-0-0} we have
\begin{equation}\label{eq:lemma-1-2}
b_{gi}
=
N^{1/2} t_jY_i( \bm z_j)/n_j - N^{-1/2} \sum_{j^\prime = 1}^J t_{j^\prime} Y_i( \bm z_{j^\prime}).
\end{equation}
The application of Lemma \ref{lemma:0} hinges on \eqref{eq:lemma-0-1}, to prove which we consider two cases:

First we discuss the case in which ``perfect co-linearity'' does not hold, i.e., there exists $j \ne j^\prime$ such that $\sigma_{j j^\prime} < \sqrt{\sigma_{jj} \sigma_{j^\prime j^\prime}}.$ On the one hand, \eqref{eq:lemma-1-2} and Cauchy-Schwartz inequality imply that
\begin{eqnarray*}
b^2_{gi}
& \le & 
2N t_j^2 Y^2_i( \bm z_j) / n_j^2 
+ \frac{2}{N} \|  \bm t \|_2^2
\sum_{j^\prime = 1}^J Y^2_i( \bm z_{j^\prime})
\\
& = & \frac{2}{N} t_j^2/\hat p_j^2 
\sum_{i^\prime=1}^N Y_{i^\prime}^2( \bm z_j)
\frac{
Y^2_i( \bm z_j)
}
{
\sum_{i^\prime=1}^N Y_{i^\prime}^2( \bm z_j)
}
+  \frac{2}{N} \|  \bm t \|_2^2
\sum_{j^\prime = 1}^J 
\left\{
\sum_{i^\prime=1}^N Y_{i^\prime}^2( \bm z_{j^\prime})
\right\}
\frac
{
Y^2_i( \bm z_{j^\prime})
}
{
\sum_{i^\prime=1}^N Y_{i^\prime}^2( \bm z_{j^\prime})
},
\end{eqnarray*}
and consequently by \eqref{eq:lemma-1-0} and Assumption \ref{assum:2}
\begin{equation}\label{eq:numerator}
\lim_{N \rightarrow \infty} \max_{1 \le g,i \le N} b_{gi}^2
\le 
2
\left(
\max_{1 \le j \le N} t_j^2 \sigma_{jj} / p_j^2
+ \| \bm t \|_2^2
\sum_{j^\prime = 1}^J \sigma_{j^\prime j^\prime}
\right)
\times 0
=
0.
\end{equation}
On the one hand, \eqref{eq:lemma-1-2} implies that
\begin{equation*}
\frac{1}{N} \sum_{g=1}^N \sum_{i=1}^N b^2_{gi} 
= \sum_{j=1}^J t_j^2 / \hat p_j 
\left\{
\frac{1}{N} \sum_{i=1}^N Y_i^2( \bm z_j)
\right\}
- \sum_{j=1}^J \sum_{j^\prime = 1}^J 
t_j t_{j^\prime} 
\left\{
\sum_{i=1}^N \frac{1}{N} Y_i( \bm z_j) Y_i( \bm z_{j^\prime})
\right\},
\end{equation*}
and consequently by Assumptions \ref{assum:1} and \ref{assum:2}
\begin{equation}\label{eq:denominator}
\lim_{N \rightarrow \infty} \frac{1}{N} \sum_{g=1}^N \sum_{i=1}^N b^2_{gi}
=
\sum_{j=1}^J t_j^2 \sigma_{jj} / p_j
- 
\sum_{j=1}^J \sum_{j^\prime = 1}^J t_j t_{j^\prime} \sigma_{jj^\prime}.
\end{equation}
We prove the right hand side of \eqref{eq:denominator} is always positive. By Cauchy-Schwartz inequality
\begin{equation}\label{eq:cauchy-inequ}
\sum_{j=1}^J t_j^2 \sigma_{jj} / p_j 
= \left( \sum_{j=1}^J t_j^2 \sigma_{jj} / p_j  \right) \left( \sum_{j=1}^J p_j \right)
\ge \sum_{j=1}^J \sum_{j^\prime = 1}^J |t_j| |t_{j^\prime}| \sqrt{\sigma_{jj} \sigma_{j^\prime j^\prime}}.
\end{equation}
Because $\bm t \ne \bm 0_J,$ the equality sign in \eqref{eq:cauchy-inequ} holds if and only if 
$
\bm t = \lambda (p_1 / \sqrt{\sigma_{11}}, \ldots, p_J / \sqrt{\sigma_{JJ}} )^\prime
$ 
for a non-zero constant $\lambda.$ Moreover, because 
$
\sigma_{j j^\prime} \le \sqrt{\sigma_{jj} \sigma_{j^\prime j^\prime}}
$ 
for all $j$ and $j^\prime,$
\begin{equation*}
\sum_{j=1}^J \sum_{j^\prime = 1}^J |t_j| |t_{j^\prime}| \sqrt{\sigma_{jj} \sigma_{j^\prime j^\prime} }  
\ge \sum_{j=1}^J \sum_{j^\prime = 1}^J t_j t_{j^\prime} \sigma_{jj^\prime}.
\end{equation*}
Additionally, the fact that there exists $j_1 \ne j_2$ such that 
$
\sigma_{j_1, j_2} < \sqrt{\sigma_{j_1, j_1} \sigma_{j_2, j_2}}
$ 
implies that if 
$
\bm t = \lambda (p_1 / \sqrt{\sigma_{11}}, \ldots, p_J / \sqrt{\sigma_{JJ}} )^\prime,
$ 
then
\begin{equation*}
\sum_{j=1}^J \sum_{j^\prime = 1}^J |t_j| |t_{j^\prime}| \sqrt{\sigma_{jj} \sigma_{j^\prime j^\prime} }  
> \sum_{j=1}^J \sum_{j^\prime = 1}^J t_j t_{j^\prime} \sigma_{jj^\prime}.
\end{equation*}
Thus we have proved that the right hand side of \eqref{eq:denominator} is positive for all $\bm t \in \mathbb R^J \backslash \{ \bm 0_J \}.$ Combining this fact with \eqref{eq:numerator}, we have proved that \eqref{eq:lemma-0-1} holds for all $\bm t \in \mathbb R^J \backslash \{ \bm 0_J \},$ and therefore \eqref{eq:normal} holds by Lemma \ref{lemma:0}.

Second, we discuss the case in which ``perfect co-linearity'' holds, i.e., $\sigma_{j j^\prime}= \sqrt{\sigma_{jj} \sigma_{j^\prime j^\prime}}$ for all $j$ and $j^\prime.$ If 
$
\bm t \ne \lambda (p_1 / \sqrt{\sigma_{11}}, \ldots, p_J / \sqrt{\sigma_{JJ}} )^\prime,
$ 
similarly as the argument for the first case
\begin{equation*}
\lim_{N \rightarrow \infty} \max_{1 \le g,i \le N} b_{gi}^2 = 0,
\quad
\lim_{N \rightarrow \infty} \frac{1}{N} \sum_{g=1}^N \sum_{i=1}^N b^2_{gi} > 0.
\end{equation*}
Therefore \eqref{eq:lemma-0-1} holds, and consequently \eqref{eq:normal} holds by Lemma \ref{lemma:0}. Otherwise, by the definition of $\bm \Sigma^\textrm{obs},$ 
\begin{equation*}
\bm t^\prime \bm \Sigma^\textrm{obs} \bm t 
= \lambda^2
\sum_{j=1}^J \left\{
(1 - p_j ) \sqrt{\sigma_{jj}}
-
\sum_{j^\prime \ne j} \frac{p_{j^\prime}\sigma_{j^\prime j}}{\sqrt{\sigma_{j^\prime j^\prime}}}
\right\}
=
\lambda^2
\sum_{j=1}^J \left\{
(1 - p_j ) \sqrt{\sigma_{jj}}
-
\sum_{j^\prime \ne j} p_{j^\prime} \sqrt{\sigma_{jj}}
\right\}
=0.
\end{equation*}
Therefore
$
\textrm{Var} (N^{1/2} \bm t^\prime \bar{ \bm Y }^\textrm{obs}) \rightarrow  0
$
by \eqref{eq:mean-variance}, and \eqref{eq:normal} holds trivially.

In summary, we have proved that \eqref{eq:normal} holds for all $\bm t \in \mathbb R^J,$ which completes the proof.
\end{proof}

\begin{lemma}\label{lemma:2}
When $N \rightarrow \infty,$
\begin{equation*}
\hat { \bm \beta }_j 
\stackrel{\mathbb P}{\longrightarrow}
\bm \zeta_j
\quad
(j=1, \ldots, J).
\end{equation*}
\end{lemma}

\begin{proof}
First, similarly as \eqref{eq:lemma-1-5}, for fixed $N$ and $k=1, \ldots, p,$
\begin{equation}\label{eq:lemma-2-1}
\textrm{Var} 
\{
\bar X_k^\textrm{obs}( \bm z_j)
\}
= \frac{1-\hat p_j}{\hat p_j} \frac{1}{N(N-1)} \sum_{i=1}^N X_{ik}^2,
\end{equation}
and
\begin{eqnarray}\label{eq:lemma-2-2}
\textrm{Var} 
\left\{ 
\frac{1}{n_j} \sum_{i=1}^N W_i( \bm z_j) X_{ik} Y_i( \bm z_j)
\right\} 
& = & \frac{1-\hat p_j}{\hat p_j} \frac{1}{N(N-1)} \sum_{i=1}^N \left\{ X_{ik}Y_i( \bm z_j) - \frac{1}{N} \sum_{i=1}^N X_{ik}Y_i( \bm z_j) \right\}^2 \nonumber \\
& \le & \frac{2(1-\hat p_j)}{(N-1) \hat p_j} \left[ \frac{1}{N} \sum_{i=1}^N X_{ik}^2Y_i^2( \bm z_j) + \frac{1}{N^2}\left\{ \sum_{i=1}^N X_{ik}Y_i( \bm z_j) \right\}^2 \right] \nonumber \\
& \le & \frac{1-\hat p_j}{\hat p_j} \frac{4L}{N-1}. 
\end{eqnarray}
The last step holds because by Cauchy-Schwartz inequality and Assumption \ref{assum:3}
\begin{equation*}
\frac{1}{N} \sum_{i=1}^N X_{ik}^2Y_i^2( \bm z_j) 
\le 
\left( \frac{1}{N} \sum_{i=1}^N X_{ik}^4 \right)^{1/2} 
\left\{ \frac{1}{N} \sum_{i=1}^N Y_i^4( \bm z_j) \right\}^{1/2}
\le L
\end{equation*}
and
\begin{equation*}
\left\{ \sum_{i=1}^N X_{ik}Y_i( \bm z_j) \right\}^2
\le 
\left( \sum_{i=1}^N X_{ik}^2 \right)
\left\{ \sum_{i=1}^N Y_i^2( \bm z_j) \right\}
\le 
N \left( \sum_{i=1}^N X_{ik}^4 \right)^{1/2} 
\left\{ \sum_{i=1}^N Y_i^4( \bm z_j) \right\}^{1/2}
\le N^2 L.
\end{equation*}

Second, by \eqref{eq:true-reg-coef} and \eqref{eq:est-reg-coef}, we only need to prove that when $N \rightarrow \infty,$
\begin{equation}\label{eq:lemma-2-3}
\frac{1}{n_j} \sum_{i=1}^N W_i( \bm z_j) \bm X_i \{Y_i( \bm z_j) - \bar Y^\textrm{obs}( \bm z_j)\}
\stackrel{\mathbb P}{\longrightarrow}
\bm \lambda_j.
\end{equation}
By \eqref{eq:lemma-1-5} and \eqref{eq:lemma-2-1} we have
$
\textrm{Var} \{ \bar Y^\textrm{obs} ( \bm z_j) \} \rightarrow 0,
$
and
$
\textrm{Var} 
\{
\bar X_k^\textrm{obs}(\bm  z_j)
\}
\rightarrow
0
$
for all
$
k.
$
Therefore by Chebyshev inequality
$
\bar Y^\textrm{obs} ( \bm z_j) \stackrel{\mathbb P}{\longrightarrow} 0
$
and
$
\bar{ X}^\textrm{obs} ( \bm z_j) \stackrel{\mathbb P}{\longrightarrow} 0.
$
By \eqref{eq:lemma-2-2}
\begin{equation*}
\textrm{Var} 
\left\{ 
\frac{1}{n_j} \sum_{i=1}^N W_i( \bm z_j) X_{ik} Y_i( \bm z_j)
\right\} 
\rightarrow 
0
\quad
(k=1, \ldots, p),
\end{equation*}
therefore
\begin{equation*}
\frac{1}{n_j} \sum_{i=1}^N W_i( \bm z_j) \bm X_i Y_i( \bm z_j) 
\stackrel{\mathbb P}{\longrightarrow}
\bm \lambda_j,
\end{equation*}
and consequently \eqref{eq:lemma-2-3} holds.
\end{proof}

\begin{lemma}\label{lemma:3} 
When N approaches infinity, in distribution
$
N^{1/2}\bar{ \bm Y }^\textrm{ca} \rightarrow N( 0,  \bm \Sigma^\textrm{ca}),
$
where
\begin{equation*}
\bm \Sigma^\textrm{ca} = 
\begin{bmatrix}
    \frac{1-p_1}{p_1} \tilde \sigma_{11} & -\tilde \sigma_{12} & \dots  & -\tilde \sigma_{1 J} \\
   -\tilde \sigma_{21} & \frac{1-p_2}{p_2} \tilde \sigma_{22} & \ldots & -\tilde \sigma_{2 J} \\
    \vdots  & \vdots & \ddots & \ldots \\
    -\tilde \sigma_{J 1}  & \dots  & \ldots & \frac{1-p_J}{p_J} \tilde \sigma_{JJ} \\
\end{bmatrix}.
\end{equation*}
\end{lemma}

\begin{proof} 
For $j=1, \ldots, J,$ let
$
\bar R^\textrm{obs}( \bm z_j) = n_j^{-1} \sum_{i=1}^N W_i( \bm z_j)R_i( \bm z_j).
$
By \eqref{eq:adj-single}
\begin{eqnarray*}
N^{1/2}\bar{ \bm Y }^\textrm{ca}
& \stackrel{\eqref{eq:adj-single}}{=} & 
\underbrace{
N^{1/2} 
\left\{
\begin{array}{c}
\bar R^\textrm{obs}( \bm z_1) \\
\vdots \\
\bar R^\textrm{obs}( \bm z_J) \\
\end{array}
\right\}
}_{ \Delta_1}
-
\underbrace{
N^{1/2}
\left\{
\begin{array}{c}
\bar{\bm X}^\textrm{obs}( \bm z_1)^\prime ( \hat{ \bm \beta}_1 - \bm \zeta_1 ) \\
\vdots\\
\bar{\bm X}^\textrm{obs}( \bm z_J)^\prime ( \hat{ \bm \beta}_J -  \bm \zeta_J ) \\
\end{array}
\right\}
}_{ \Delta_2}.
\end{eqnarray*}
On the one hand, $R_i( \bm z_j)$'s satisfy Assumption \ref{assum:3}, because by Cauchy-Schwartz inequality
\begin{equation*}
R_i( \bm z_j)^4
\leq
(p+1)^3 
\left\{
Y_i( \bm z_j)^4
+ \left( \max_{1 \le j \le J} \|  \bm \zeta_j \|_\infty^4 \right) \sum_{k=1}^p X_{ik}^4
\right\}.
\end{equation*}
By substituting $Y_i( \bm z_j)$ with $R_i( \bm z_j)$ and applying Lemma \ref{lemma:1}, we have
$
\Delta_1  
\stackrel{\mathbb D}{\longrightarrow}
N( 0,  \bm \Sigma^\textrm{ca}).   
$
On the other hand, by Lemma \ref{lemma:2} we have
$
N^{1/2} \bar{ \bm X}^\textrm{obs}( \bm z_j) = O_{P}( 1)
$
and
$
\hat{ \bm \beta}_j -  \bm \zeta_j = o_{P}( \bm 1),
$
which implies that
$
\Delta_2 
\stackrel{\mathbb P}{\longrightarrow}
0,
$
and by Slutsky Theorem
$
\Delta_1 -  \Delta_2 \stackrel{\mathbb D}{\longrightarrow} N( 0, \bm \Sigma^\textrm{ca}).
$
\end{proof}

\subsection{Main Results}

With the help of Lemmas \ref{lemma:0}--\ref{lemma:3}, we now state and prove the main results.

\begin{theorem}\label{prop:main}
The randomization-based and covariate-adjusted estimators are both asymptotically normal, i.e.,
\begin{equation*}
N^{1/2} \left\{ \hat \tau_\textrm{rb}(l) - \tau(l) \right\}  
\stackrel{\mathbb D}{\longrightarrow} 
N \{ 0, \sigma_\textrm{rb}^2(l) \},
\quad
N^{1/2} \left\{ \hat \tau_\textrm{ca}(l) - \tau(l) \right\} 
\stackrel{\mathbb D}{\longrightarrow} 
N \{ 0, \sigma_\textrm{ca}^2(l) \},
\end{equation*}
where
\begin{equation}\label{eq:avar-est}
\sigma_\textrm{rb}^2 (l)
=
\frac{1}{2^{2(K-1)}}
\left( 
\sum_{j=1}^J \frac{1-p_j}{p_j}\sigma_{jj} - \sum_{j \ne j^\prime} h_{jl}h_{j^\prime l} \sigma_{j j^\prime}
\right)
\end{equation}
and
\begin{equation}\label{eq:avar-adj}
\sigma_\textrm{ca}^2 (l)
=
\frac{1}{2^{2(K-1)}}
\left( 
\sum_{j=1}^J \frac{1-p_j}{p_j} \tilde \sigma_{jj} - \sum_{j \ne j^\prime} h_{jl}h_{j^\prime l} \tilde \sigma_{j j^\prime}
\right).
\end{equation}
\end{theorem}

\begin{proof}[of Theorem \ref{prop:main}]
The asymptotically normality of $\hat \tau_\textrm{rb}(l)$ follows from the fact that it is a linear combination of $\bar{ \bm Y }^\textrm{obs},$ which by Lemma \ref{lemma:1} is asymptotically multivariate normal. Moreover, \eqref{eq:avar-est} holds by \eqref{eq:est}. We apply similar argument to $\hat \tau_\textrm{ca}(l),$ in which we use Lemma \ref{lemma:3}.
\end{proof}

\begin{corollary}\label{coro:1}
Let
\begin{equation}\label{eq:coro-1-0}
\bm \xi_{j j^\prime} = \left( \frac{p_{j^\prime}}{p_j} \right)^{1/2} h_{jl}  \bm \zeta_j - \left( \frac{p_j}{p_{j^\prime}}\right)^{1/2} h_{j^\prime l} \bm \zeta_{j^\prime}
\quad
(j,j^\prime = 1, \ldots, J).
\end{equation}
The difference of the asymptotic precisions between the randomization-based estimator and the covariate-adjusted estimator is
\begin{equation}\label{eq:coro-1-main}
\mathrm{Var} \{ \hat \tau_\textrm{rb}(l) \} - \mathrm{Var} \{ \hat \tau_\textrm{ca}(l) \}
= \frac{1}{2^{2K-1}N} \sum_{j=1}^J \sum_{j^\prime=1}^J  \bm \xi_{j j^\prime}^\prime  \bm \Omega \bm \xi_{j j^\prime}.
\end{equation}
\end{corollary}

\begin{proof}[Proof of Corollary \ref{coro:1}]
On the one hand, by \eqref{eq:coro-1-0}
\begin{equation}\label{eq:coro-1-1}
\bm \xi_{j j^\prime}^\prime \bm \Omega \bm \xi_{j, j^\prime} =
\frac{p_{j^\prime}}{p_j} \bm \zeta_j^\prime \bm \Omega \bm \zeta_j
+
\frac{p_j}{p_{j^\prime}} \bm \zeta_{j^\prime}^\prime \bm \Omega \bm \zeta_{j^\prime}
-
h_{jl}h_{j^\prime l} \bm \zeta_j^\prime \bm \Omega \bm \zeta_{j^\prime}
-
h_{j^\prime l}h_{jl} \bm \zeta_{j^\prime}^\prime \bm \Omega \bm \zeta_j.
\end{equation}
On the other hand, by \eqref{eq:avar-est} and \eqref{eq:avar-adj}
\begin{eqnarray*}
2^{2(K-1)} \{ \sigma_\textrm{rb}^2 (l) - \sigma_\textrm{ca}^2 (l) \}
& = &  
\sum_{j=1}^J \frac{1-p_j}{p_j} ( \sigma_{jj} - \tilde \sigma_{jj} ) - \sum_{j\ne j^\prime}  h_{jl}h_{j^\prime l} (\sigma_{jj^\prime} - \tilde \sigma_{jj^\prime} ) \\
& = &  
\sum_{j=1}^J \frac{1-p_j}{p_j}  \bm \lambda_j^\prime \bm \Omega^{-1}  \bm \lambda_j - \sum_{j \ne j^\prime}  h_{jl}h_{j^\prime l}  \bm \lambda_j^\prime \bm \Omega^{-1} \bm \lambda_{j^\prime} \\
& = &
\sum_{j=1}^J \frac{1}{p_j}  \bm \zeta_j^\prime \bm \Omega \bm \zeta_j \sum_{j^\prime = 1}^J p_{j^\prime} - \sum_{j=1}^J \sum_{j^\prime=1}^J  h_{jl}h_{j^\prime l} \bm \zeta_j^\prime \bm \Omega \bm \zeta_{j^\prime} \\
& = & \frac{1}{2} \sum_{j=1}^J \sum_{j^\prime=1}^J \bm \xi_{j j^\prime}^\prime \bm \Omega \bm \xi_{j j^\prime}.
\end{eqnarray*}
The last equation holds by \eqref{eq:coro-1-1}.
\end{proof}

Theorem \ref{prop:main} illustrates the asymptotic unbiasedness and consistency of the randomization-based estimator and the covariate-adjusted estimator, and Corollary \ref{coro:1} illustrates the asymptotic precision by performing covariate adjustment. In particular, covariate adjustment never hurts asymptotic precision, and by \eqref{eq:coro-1-main} the sufficient and necessary condition for the randomization-based estimator and the covariate-adjusted estimator to be asymptotically equally precise is
\begin{equation*}
p_{j^\prime} h_{jl} \bm \zeta_j =  p_jh_{j^\prime l} \bm \zeta_{j^\prime}
\quad
(j, j^\prime = 1, \ldots, J).
\end{equation*}

\section{Concluding Remarks}\label{sec:discuss}

In this paper, we define the covariate-adjusted estimator for $2^K$ factorial designs, and derive the asymptotic precisions of the unadjusted and covariate-adjusted estimators. We confirm that both the unadjusted and covariate-adjusted estimators are asymptotically unbiased and normal, and the latter is more precise than the former. Moreover, we quantify the precision gained by performing covariate adjustment.

Our work implies multiple future directions. First, we can generalize our current framework to other factorial designs such as $3^k$ factorial designs or fractional factorial designs. Second, it is necessary to investigate the finite-sample properties of the estimators. In particular, although the covariate-adjusted estimator is asymptotically unbiased, it is biased from a finite-sample perspective. \cite{Lin:2013} showed that for randomized treatment-control studies the finite-sample bias of the covariate-adjusted estimator is $O(N^{-1}),$ and it would be helpful to generate this result to factorial designs. Moreover, \cite{Lu:2016} showed that for $2^K$ factorial designs we can adopt the amended Huber-White sandwich estimator HC2 \citep{MacKinnon:1985} for estimating the sampling variance of the unadjusted estimator, and therefore it would be helpful to have parallel results for the covariate-adjusted estimator. Third, it is possible to incorporate Bayesian analysis into our current framework.

\section*{Acknowledgement}

The author thanks several colleagues at the Microsoft Analysis and Experimentation Team, especially Randy Henne and Alex Deng, for inspiring this project, and Professor Tirthankar Dasgupta at Harvard University and Professor Peng Ding at UC Berkeley for helpful suggestions. Thoughtful comments from the Co-Editor-in-Chief, an Association Editor and a reviewer have substantially improved the quality of this paper.

\bibliographystyle{apalike}
\bibliography{factorial_covariate}

\begin{thebibliography}{}

\bibitem[Bailey, 1981]{Bailey:1981}
Bailey, R.~A. (1981).
\newblock A unified approach to design of experiments.
\newblock {\em Journal of the Royal Statistical Society: Series A},
  144(2):214--223.

\bibitem[Bailey, 1991]{Bailey:1991}
Bailey, R.~A. (1991).
\newblock Strata for randomized experiments (with discussion).
\newblock {\em Journal of the Royal Statistical Society: Series B},
  53(1):27--78.

\bibitem[Cochran, 1977]{Cochran:1977}
Cochran, W.~G. (1977).
\newblock {\em Sampling Techniques, 3rd Edition}.
\newblock New York: W.W. Norton.

\bibitem[Dasgupta et~al., 2015]{Dasgupta:2015}
Dasgupta, T., Pillai, N., and Rubin, D.~B. (2015).
\newblock Causal inference from $2^k$ factorial designs using the potential
  outcomes model.
\newblock {\em Journal of the Royal Statistical Society: Series B},
  77(4):727--753.

\bibitem[Deng et~al., 2013]{Deng:2013}
Deng, A., Xu, Y., Kohavi, R., and Walker, T. (2013).
\newblock Improving the sensitivity of online controlled experiments by
  utilizing pre-experiment data.
\newblock In {\em Proceedings of the 6th ACM international conference on web
  search data mining}, pages 123--132.

\bibitem[Espinosa et~al., 2016]{Espinosa:2016}
Espinosa, V., Dasgupta, T., and Rubin, D.~B. (2016).
\newblock A {B}ayesian perspective on the analysis of unreplicated factorial
  experiments using potential outcomes.
\newblock {\em Technometrics}, 58:62--73.

\bibitem[Fisher, 1935]{Fisher:1935}
Fisher, R.~A. (1935).
\newblock {\em The Design of Experiments}.
\newblock Edinburgh: Oliver and Boyd.

\bibitem[Freedman, 2008]{Freedman:2008}
Freedman, D.~A. (2008).
\newblock On regression adjustments in experiments with several treatments.
\newblock {\em The Annals of Applied Statistics}, 2(1):176--196.

\bibitem[Hoeffding, 1951]{Hoeffding:1951}
Hoeffding, W. (1951).
\newblock A combinatorial central limit theorem.
\newblock {\em The Annals of Mathematical Statistics}, 22(4):558--566.

\bibitem[Kempthrone, 1952]{Kempthrone:1952}
Kempthrone, O. (1952).
\newblock {\em The Design and Analysis of Experiments}.
\newblock New York: Wiley.

\bibitem[Kempthrone, 1955]{Kempthrone:1955}
Kempthrone, O. (1955).
\newblock The randomization theory of experimental inference.
\newblock {\em Journal of American Statistical Association}, 50(271):946--967.

\bibitem[Lin, 2013]{Lin:2013}
Lin, W. (2013).
\newblock Agnostic notes on regression adjustments to experimental data:
  Reexamining freedman's critique.
\newblock {\em The Annals of Applied Statistics}, 7(1):295--318.

\bibitem[Lu, 2016]{Lu:2016}
Lu, J. (2016).
\newblock On randomization-based and regression-based inferences for $2^k$
  factorial designs.
\newblock {\em Statistics and Probability Letters}, 112:72--78.

\bibitem[Lu et~al., 2015]{Lu:2015}
Lu, J., Ding, P., and Dasgupta, T. (2015).
\newblock Construction of alternative hypotheses for randomization tests with
  ordinal outcomes.
\newblock {\em Statistics and Probability Letters}, 107:348--355.

\bibitem[MacKinnon and White, 1985]{MacKinnon:1985}
MacKinnon, J.~G. and White, H. (1985).
\newblock Some heteroskedasticity-consistent covariance matrix estimators with
  improved finite sample properties.
\newblock {\em Journal of Econometrics}, 29(3):305--325.

\bibitem[Miller, 2006]{Miller:2006}
Miller, S. (2006).
\newblock {\em Experimental Design and Statistics}.
\newblock London: Routledge.

\bibitem[Miratrix et~al., 2013]{Miratrix:2013}
Miratrix, L.~W., Sekhon, S.~S., and Yu, B. (2013).
\newblock Adjusting treatment effect estimates by post-stratification in
  randomized experiments.
\newblock {\em Journal of the Royal Statistical Society: Series B},
  75(2):369--396.

\bibitem[Neyman, 1923]{Neyman:1923}
Neyman, J.~S. (1990[1923]).
\newblock On the application of probability theory to agricultural experiments.
  essay on principles (with discussion). section 9 (translated). reprinted ed.
\newblock {\em Statistical Science}, 5(4):465--472.

\bibitem[Rubin, 1974]{Rubin:1974}
Rubin, D.~B. (1974).
\newblock Estimating causal effects of treatments in randomized and
  nonrandomized studies.
\newblock {\em Journal of Educational Psuchology}, 66(5):688--701.

\bibitem[Rubin, 1980]{Rubin:1980}
Rubin, D.~B. (1980).
\newblock Comment on ``{R}andomized analysis of experimental data: the fisher
  randomization test'' by {D}. {B}asu.
\newblock {\em Journal of American Statistical Association}, 75(371):591--593.

\bibitem[Rubin, 2008]{Rubin:2008}
Rubin, D.~B. (2008).
\newblock For objective causal inference, design trumps analysis.
\newblock {\em The Annals of Applied Statistics}, 2(3):808--840.

\bibitem[Wilk and Kempthrone, 1956]{Wilk:1956}
Wilk, M.~B. and Kempthrone, O. (1956).
\newblock Some aspects of the analysis of factorial experiments in a completely
  randomized design.
\newblock {\em The Annals of Mathematical Statistics}, 27(4):950--985.

\bibitem[Wu and Hamada, 2009]{Wu:2009}
Wu, C. F.~J. and Hamada, M.~S. (2009).
\newblock {\em Experiments: Planning, Analysis, and Optimization}.
\newblock New York: Wiley.

\bibitem[Yates, 1937]{Yates:1937}
Yates, F. (1937).
\newblock The design and analysis of factorial experiments.
\newblock {\em Technical Communication}, 35.
\newblock Imperial Bureau of Soil Science, London.

\end{thebibliography}

\end{document}